\documentclass[copyright,creativecommons]{eptcs}
\providecommand{\event}{GandALF 2013}   
\usepackage{breakurl} 

\usepackage[latin1]{inputenc}
\usepackage[T1]{fontenc}
\usepackage{enumerate}
\usepackage{amsthm}
\usepackage{amsmath}
\usepackage{amsfonts}

\usepackage{xcolor,pgf,tikz,pgflibraryarrows,pgffor,pgflibrarysnakes}
\usetikzlibrary{shapes,decorations.markings}

\newcommand{\N}{\mathbb N}

\newtheorem{theorem}{Theorem}[section]

\newtheorem{corollary}[theorem]{Corollary}
\newtheorem{proposition}[theorem]{Proposition}

\theoremstyle{definition}
\newtheorem{definition}[theorem]{Definition}
\newtheorem{example}[theorem]{Example}

\theoremstyle{remark}
\newtheorem{remark}[theorem]{Remark}

\numberwithin{equation}{section}

\title{Simple strategies for Banach-Mazur games\\ and fairly correct
  systems \thanks{This work has been partly supported by a grant from
    the National Bank of Belgium, the ARC project (number
    AUWB-2010-10/15-UMONS-3), and the FRFC project (number
    2.4545.11).}}
\author{Thomas Brihaye \institute{UMONS\\ Mons, Belgium}
  \institute{University of Mons\\ Place du Parc 20, 7000 Mons,
    Belgium} \email{thomas.brihaye@umons.ac.be} \and Quentin Menet
\thanks{The second author is supported by a grant of FRIA.}
 \institute{UMONS\\ Mons, Belgium}
 \institute{University of Mons\\
  Place du Parc 20, 7000 Mons, Belgium} 
 \email{quentin.menet@umons.ac.be}
 }

\def\titlerunning{Simple strategies for Banach-Mazur games and fairly correct systems}
\def\authorrunning{Thomas Brihaye \& Quentin Menet}

\begin{document}
\maketitle

\begin{abstract}
In 2006, Varacca and V{\"o}lzer proved that on finite graphs,
$\omega$-regular large sets coincide with $\omega$-regular sets of
probability~1, by using the existence of positional strategies in the
related Banach-Mazur games. Motivated by this result, we try to
understand relations between sets of probability~1 and various
notions of simple strategies (including those introduced in a recent
paper of Gr{\"a}del and Le{\ss}enich).
Then, we introduce a generalisation of the
classical Banach-Mazur game and in particular, a probabilistic version
whose goal is to characterise sets of probability~1 (as classical
Banach-Mazur games characterise large sets). We obtain a determinacy
result for these games, when the winning set is a countable
intersection of open sets.
\end{abstract}

\section{Introduction}
Systems (automatically) controlled by computer programs abound in our
everyday life.
Clearly enough, it is of a capital importance to know whether the
programs governing these systems are \emph{correct}. Over the last
thirty years, formal methods for verifying computerised systems have
been developed for validating the adequation of the systems against
their requirements. Model checking is one such approach: it consists
first in modelling the system under study (for instance by an
automaton), and then in applying algorithms for comparing the
behaviours of that model against a specification (modelled for instance
by a logical formula). Model checking has now reached maturity,
through the development of efficient symbolic techniques,
state-of-the-art tool support, and numerous successful applications to
various areas.

As argued in~\cite{VV06}: \emph{`Sometimes, a model of a concurrent or
  reactive system does not satisfy a desired linear-time temporal
  specification but the runs violating the specification seem to be
  artificial and rare'}.  As a naive example of this phenomenon,
consider a coin flipped an infinite number of times. Classical
verification will assure that the property stating \emph{``one day, we
  will observe at least one head''} is false, since there exists a unique
execution of the system violating the property. In some situations,
for instance when modeling non-critical systems, one could prefer to
know whether the system is \emph{fairly correct}. Roughly speaking, a
system is fairly correct against a property if the set of executions
of the system violating the property is \emph{``very small''}; or
equivalently if the set of executions of the system satisfying the
property is \emph{``very big''}. A first natural notion of fairly
correct system is related to probability: \emph{almost-sure
  correctness}. A system is almost-surely correct against a property
if the set of executions of the system satisfying the property has
probability~$1$. Another interesting notion of fairly correct system
is related to topology: \emph{large correctness}. A system is largely
correct against a property if the set of executions of the system
satisfying the property is \emph{large} (in the topological
sense). There exists a lovely characterisation of \emph{large sets} by
means of the \emph{Banach-Mazur games}. In~\cite{oxtoby57}, it has
been shown that a set $W$ is large if and only if a player has a
winning strategy in the related Banach-Mazur game.

Although, the two notions of \emph{fairly correct systems} do not
coincide in general, in~\cite{VV06}, the authors proved (amongst
other results) the following result: when considering $\omega$-regular
properties on finite systems, the \emph{almost-sure correctness} and
the \emph{large correctness} coincide, for bounded Borel measures. Motivated by this very nice result, we intend to
extend it to a larger class of specifications.  The key ingredient to
prove the previously mentioned result of~\cite{VV06} is that when
considering $\omega$-regular properties, \emph{positional} strategies
are sufficient in order to win the related Banach-Mazur
game~\cite{BGK03}. For this reason, we investigate \emph{simple
  strategies} in Banach-Mazur games, inspired by the recent
work~\cite{GL12} where infinite graphs are studied.


\textbf{Our contributions.}  In this paper, we first compare various
notions of simple strategies on finite graphs (including
\emph{bounded} and \emph{move-counting} strategies), and their
relations with the sets of probability~$1$. Given a set $W$, the
existence of a bounded (resp. move-counting) winning strategy in the
related Banach-Mazur game implies that $W$ is a set of
probability~$1$. However there exist sets $W$ of probability~$1$ for
which there is no bounded and no move-counting winning strategy in
the related Banach-Mazur game.
Therefore, we introduce a generalisation of the classical Banach-Mazur
game and in particular, a probabilistic version whose goal is to
characterise sets of probability~1 (as classical Banach-Mazur games
characterise large sets). We obtain the desired characterisation in
the case of countable intersections of open sets. This is the main
contribution of the paper. As a byproduct of the latter, we get a
determinacy result for our probabilistic version of the Banach-Mazur game
for countable intersections of open sets.






\section{Banach-Mazur Games on finite graphs}\label{finite}

Let $(X,\mathcal{T})$ be a topological space. A notion of topological
``bigness'' is given by large sets.  A subset $W\subset X$ is said to
be \emph{nowhere dense} if the closure of $W$ has empty interior.  A
subset $W\subset X$ is said to be \emph{meagre} if it can be expressed
as the union of countably many nowhere dense sets and a subset
$W\subset X$ is said to be large if $W^c$ is meagre. In particular, we
remark that a countable intersection of large sets is still large and
that if $W\subset X$ is large, then any set $Y\supset W$ is large.

If $G=(V,E)$ is a finite directed graph and $v_0\in V$, then the space of infinite
paths in $G$ from $v_0$, denoted $\text{Paths}(G,v_0)$, can be endowed
with the complete metric
%
\begin{equation}\label{metricgraph}
{d((\sigma_n)_{n\geq 0},(\rho_n)_{n\geq 0})=2^{-k}} \quad\text{where}\quad k=\min\{n\ge
0: \sigma_n\ne \rho_n\}
\end{equation} 
with the conventions that $\min \emptyset=\infty$ and
$2^{-\infty}=0$. 
In other words, the open sets in $\text{Paths}(G,v_0)$
endowed with this metric
are the countable unions of cylinders, where
a cylinder is a set of the form
${\{\rho \in
\text{Paths}(G,v_0) ~|~ \pi \text{ is a prefix of } \rho\}}$
for some finite path $\pi$ in $G$ from $v_0$.

We can therefore study the large subsets of the metric
space $(\text{Paths}(G,v_0),d)$.  Banach-Mazur games allow us to
characterise large subsets of this metric space through the existence
of winning strategies.

\begin{definition}
A \emph{Banach-Mazur game} $\mathcal{G}$ on a finite graph is a
triplet $(G,v_0,W)$ where $G=(V,E)$ is a finite directed graph where
every vertex has a successor, $v_0\in V$ is the initial state, $W$ is
a subset of the infinite paths in $G$ starting in~$v_0$.
\end{definition}

A Banach-Mazur game $\mathcal{G}=(G,v_0,W)$ on a finite graph is a
two-player game where Pl.~0 and Pl.~1 alternate in choosing a finite
path as follows: Pl. 1 begins with choosing a finite\footnote{In this
  paper, we always assume that a finite path is non-empty.} path
$\pi_1$ starting in $v_0$; Pl.~0 then prolongs $\pi_1$ by choosing
another finite path $\pi_2$ and so on. A play of $\mathcal{G}$ is thus
an infinite path in $G$ and we say that Pl. 0 wins if this path
belongs to $W$, while Pl.~1 wins if this path does not belong to
$W$. The set $W$ is called the winning condition.  It is important to
remark that, in general, in the literature, Pl.~0 moves first in
Banach-Mazur games but in this paper, we always assume that Pl.~1
moves first in order to bring out the notion of large set (rather than
meagre set).  The main result about Banach-Mazur games can then be
stated as follows:

\begin{theorem}[\cite{oxtoby57}]\label{thm BMfinite}
Let $\mathcal{G}=(G,v_0,W)$ be a Banach-Mazur game on a finite graph.
Pl. 0 has a winning strategy for $\mathcal{G}$ if and only if $W$ is
large.
\end{theorem}

\section{Simple strategies in Banach-Mazur games}

In a Banach-Mazur game $(G,v_0,W)$ on a finite graph, a strategy for
Pl. 0 is given by a function $f$ defined on $\text{FinPaths}(G,v_0)$,
the set of finite paths of $G$ starting from $v_0$, such that for any
$\pi\in \text{FinPaths}(G,v_0)$, we have $f(\pi)\in
\text{FinPaths}(G,\text{last}(\pi))$. However, we can imagine some
restrictions on the strategies of Pl.~0:
\begin{enumerate}
\item A strategy $f$ is said to be \emph{positional} if it only
  depends on the current vertex, i.e $f$ is a function defined on $V$
  such that for any $v\in V$, $f(v)\in \text{FinPaths}(G,v)$
  and a play $\rho$ is consistent with $f$ if 
  $\rho$ is of the form $(\pi_i f(\text{last}(\pi_i))_{i\ge 1}$.
\item A strategy $f$ is said to be \emph{finite-memory} if it only
  depends on the current vertex and a finite memory
  (see~\cite{Gr08} for the precise definition of a finite-memory
  strategy).
\item A strategy $f$ is said to be \emph{b-bounded} if for any $\pi\in
  \text{FinPaths}(G,v_0)$, $f(\pi)$ has length less than $b$ and a
  strategy is said to be \emph{bounded} if there is $b\ge 1$ such that
  $f$ is b-bounded.
\item A strategy $f$ is said to be \emph{move-counting} if it only
  depends on the current vertex and the number of moves already played,
  i.e. $f$ is a function defined on $V\times \mathbb{N}$ such that for
  any $v\in V$, any $n\in \mathbb{N}$, $f(v,n)\in
  \text{FinPaths}(G,v)$ and a play $\rho$ is consistent with $f$ if $\rho$ 
  is of the form $(\pi_i f(\text{last}(\pi_i),i))_{i\ge 1}$.
\item A strategy $f$ is said to be \emph{length-counting} if it only
  depends on the current vertex and the length of the prefix already
  played, i.e. $f$ is a function defined on $V\times \mathbb{N}$ such
  that for any $v\in V$, any $n\in \mathbb{N}$, $f(v,n)\in
  \text{FinPaths}(G,v)$ and a play $\rho$ is consistent with $f$ if after
  a prefix $\pi$, the move of Pl.~0 is given by $f(\text{last}(\pi),|\pi|)$.
\end{enumerate}

The notions of positional and finite memory strategies are classical,
bounded strategies are present in~\cite{VV06}, move-counting and
length-counting strategies have been introduced in~\cite{GL12}.  We
first remark that, by definition, the existence of a positional
winning strategy implies the existence of finite-memory/move-counting/length-counting
 winning strategies. Moreover, since $G$
is a finite graph, a positional strategy is always
bounded. In~\cite{Gr08}, it is proved that the existence of a
finite-memory winning strategy implies the existence of a positional
winning strategy.

\begin{proposition}[\cite{Gr08}]
Let $\mathcal{G}=(G,v_0,W)$ be a Banach-Mazur game. Pl.~0 has a
finite-memory winning strategy if and only if Pl.~0 has a positional
winning strategy.
 \end{proposition}

Using the ideas of the proof of the above proposition, we can also
show that the existence of a winning strategy implies the existence of
a length-counting winning strategy.

\begin{proposition}\label{equivlength}
Let $\mathcal{G}=(G,v_0,W)$ be a Banach-Mazur game on a finite graph. 
Pl.~0 has a length-counting winning strategy if and only if Pl.~0 has a winning strategy.
\end{proposition}
\begin{proof}
Let $f$ be a winning strategy for Pl.~0. Since $G$ is a
finite graph, for any $n\ge 0$ and any $v\in
V$, we can consider an enumeration $\pi_1,\dots,\pi_m$ of finite paths in
$\text{FinPaths}(G,v_0)$ of length $n$ such that
$\text{last}(\pi_i)=v$. We then let
\[h(v,n)=f\big(\pi_1\big)f\big(\pi_2f(\pi_1)\big)
f\big(\pi_3 f(\pi_1)f(\pi_2f(\pi_1))\big) \dots
f\big(\pi_mf(\pi_1)f(\pi_2f(\pi_1))\cdots\big).\] If $\rho$ is a play
consistent with $h$, then $\rho$ is a play where the strategy $f$ is
applied infinitely often. Thus such a play $\rho$ can be seen as a
play $\sigma_1 \tau_1 \sigma_2 \tau_2 \cdots$ where the $\tau_i$'s
(resp. the $\sigma_i$'s) are the moves of Pl.~0 (resp. Pl.~1.) and
where $f(\sigma_1 \tau_1 \cdots \sigma_i)=\tau_i$.  Each play
consistent with $h$ can thus be seen as a play consistent with $f$,
and we deduce that the strategy $h$ is a length-counting winning
strategy.
\end{proof}

On the other side, the notions of move-counting winning strategies and
bounded winning strategies are incomparable.

\begin{example}[\textbf{Set with a move-counting winning strategy and without a bounded winning strategy}]\label{nobound}
We consider the complete graph $G_{0,1}$ on $\{0,1\}$. Let $W$ be the
set of any sequences $(\sigma_n)_{n\ge 1}$ in $\{0,1\}^{\omega}$ with
$\sigma_1=0$ such that $(\sigma_n)_{n\ge 1}$ contains a finite sequence of $1$
strictly longer than the initial finite sequence of $0$. In other words,
$(\sigma_n)_{n\ge 1}\in W$ if $\sigma_1=0$ and if there exist $j\ge 1$ and $k\ge 1$ such that
$\sigma_{j}=1$ and $\sigma_{k+1}=\cdots=\sigma_{k+j}=1$. 
Let $\mathcal{G}=(G_{0,1},0,W)$. The strategy $f(\cdot,n)= 1^n$ is a move-counting winning
strategy for Pl.~0 for the game $\mathcal{G}$. On the other hand,
there does not exist a bounded winning
strategy for Pl.~0 for the game $\mathcal{G}$. Indeed, if $f$ is a $b$-bounded strategy of Pl.~0, then Pl.~1 can start
 by playing $0^b$ and then, always play $0$.
\end{example}

\begin{example}[\textbf{Set with a bounded winning strategy and without a move-counting winning strategy}]\label{nomove}
We consider the complete graph $G_{0,1}$ on $\{0,1\}$. Let
$(\pi_n)_{n\ge 0}$ be an enumeration of $\text{FinPaths}(G)$ with
$\pi_0=0$. We let $W$ be the set of any sequences in
$\{0,1\}^{\omega}$ starting by $0$ except the sequence
$\rho=\pi_0\pi_1\pi_2\dots$. Let $\mathcal{G}=(G_{0,1},0,W)$. It is
obvious that Pl.~0 has a $1$-bounded winning strategy for
$\mathcal{G}$ but we can also prove that Pl.~0 has no move-counting
winning strategy.  Indeed, if $h$ is a move-counting strategy of
Pl.~0, then Pl.~1 can start by playing a prefix $\pi$ of $\rho$ so
that $\pi h(\text{last}(\pi),1)$ is a prefix of $\rho$. Afterwards,
Pl.~1 can play $\pi'$ such that $\pi h(\text{last}(\pi),1)\pi'h(\text{last}(\pi'),2)$  is a prefix of $\rho$ and so on.
\end{example}
 
We remark that the sets $W$ considered in these examples are
\emph{open} sets, i.e. sets on a low level of the Borel
hierarchy. Moreover, by Proposition~\ref{equivlength}, there also
exist length-counting winning strategies for these two examples. The
relations between the simple strategies are thus completely
characterised and are summarised in Figure~\ref{resume2}. This
Figure also contains other simple strategies which will be discussed
later.

\section{Link with the sets of probability $1$}
Let $G=(V,E)$ be a finite directed graph. We can easily define a
probability measure $P$, on the set of infinite paths in $G$, by
giving a weight $w_e>0$ at each edge $e\in E$ and by considering that
for any $v, v' \in V$, $p_w(v,v')=0$ if $(v,v') \not\in E$ and
$p_w(v,v')=\frac{w_{(v,v')}}{\sum_{e' \text{\ enabled from v\ }}w_{e'}}$
else, where $p_w(v,v')$ denotes the probability of taking edge $(v,v')$
from state $v$. Given $v_1 \cdots v_n \in \text{FinPaths}(G,v_1)$, we
recall that we denote by $\text{Cyl}(v_1 \cdots v_n)$ the cylinder generated by $v_1
\cdots v_n$ and defined as $\text{Cyl}(v_1 \cdots v_n) = \{\rho \in
\text{Paths}(G,v_1) ~|~ v_1\cdots v_n \text{ is a prefix of } \rho\}$.
\begin{definition}
Let $G=(V,E)$ be a finite directed graph and $w=(w_e)_{e\in E}$ a
family of positive weights.  We define the probability measure $P_w$
by the relation
\begin{equation}
\label{eqP}
P_w(\text{{Cyl}}(v_1\cdots v_n)) =p_w(v_1,v_2) \cdot \dots \cdot p_w(v_{n-1},v_n)
\end{equation}
and we say that such a probability measure is \emph{reasonable}.
\end{definition}

We are interested in characterising the sets $W$ of probability~$1$
and their links with the different notions of simple winning
strategies. We remark that, in general, Banach-Mazur games do not
characterise sets of probability $1$. In other words, the notions of
large sets and sets of probability~$1$ do not coincide in general on finite graphs.
Indeed, there exist some large sets of probability~$0$. We
present here an example of such sets:

\begin{example}[\textbf{Large set of probability $0$}]\label{wwRdef}
We consider the complete graph $G_{0,1,2}$ on $\{0,1,2\}$ and the set
$W=\{(w_iw^R_i)_{i\geq 0}\in \text{Paths}(G_{0,1,2},2):w_i\in\{0,1,2\}^{*}\}$, where for any finite word
$\sigma\in\{0,1,2\}^{*}$ given by $\sigma=\sigma(1)\cdots\sigma(n)$
with $\sigma(i)\in\{0,1,2\}$, we let
$\sigma^R=\sigma(n)\cdots\sigma(1)$.
In other words, $W$ is the set of runs $\rho$ starting from $2$ that we can divide into a consecutive sequence of finite words and their reverse. It is obvious that Pl.~0 has a winning strategy for the Banach-Mazur game
$(G_{0,1,2},2,W)$ and thus that $W$ is large. On the other hand, if $P$ is the reasonable probability measure given by the
weights $w_e=1$ for any $e\in E$, then we can verify that
$P(W)=0$. 
Indeed, we have
\begin{align*}
P(W)&\le \sum_{n=1}^{\infty}P(\{w_0w^R_0(w_iw^R_i)_{i\geq 1}\in W: |w_0|=n\})\\
&= \sum_{n=1}^{\infty} P(\{w_0w^R_0w\in \text{Paths}(G_{0,1,2},2):|w_0|=n\})\cdot P(W)\\
&\le \sum_{n=1}^{\infty}\frac{P(W)}{3^n}= \frac{1}{2}P(W).
\end{align*}
\end{example}

For certain families of sets, we can however have an equivalence
between the notion of large set and the notion of set of probability
$1$. It is the case for the family of sets $W$ representing
$\omega$-regular properties on finite graphs (see \cite{VV06}). In
order to prove this equivalence for $\omega$-regular sets, Varacca and
V{\"o}lzer have in fact used the fact that for these sets, the
Banach-Mazur game is positionally determined (\cite{BGK03}) and that
the existence of a positional winning strategy for Pl.~0 implies
$P(W)=1$. This latter assertion follows from the fact that every
positional strategy is bounded and that, by the Borel-Cantelli lemma,
the set of plays consistent with a bounded strategy is a set of
probability $1$.  Nevertheless, if $W$ does not represent an
$\omega$-regular properties, it is possible that $W$ is a large set of
probability~$1$ and that there is no positional winning strategy for
Pl.~0 and even no bounded or move-counting winning strategy.
 
\begin{example}[\textbf{Large set of probability $1$ without a positional/ bounded/ move-counting winning strategy}]\label{ex:pos}
We consider the complete graph $G_{0,1}$ on $\{0,1\}$ and the
reasonable probability measure $P$ given by $w_e=1$ for any $e\in E$. Let
$a_n=\sum_{k=1}^nk$. We let $W=\{(\sigma_k)_{k\ge
  1}\in\{0,1\}^{\omega}: \sigma_1=0 \ \text{and}\ \sigma_{a_n}=1
\ \text{for some}\ n > 1\}$ and $\mathcal{G}=(G_{0,1},0,W)$. Since
Pl.~0 has a winning strategy for $\mathcal{G}$, we deduce that $W$ is
a large set. We can also compute that $P(W)=1$ because if we denote by
$A_n$, $n>1$, the set
\[A_n:=\{(\sigma_k)_{k\ge 1}\in\{0,1\}^{\omega}: \sigma_{a_n}=1 \ \text{and}\ \sigma_{a_m}=0 \ \text{for any}\ m<n\},\] we have:
\[W=\dot{\bigcup}_{n> 1}A_n \quad \text{and}
\quad P(A_n)=\frac{1}{2^{n-1}}.\] 
On the other hand, there does not 
exist any positional (resp. bounded) winning strategy $f$ for
Pl.~0. Indeed, if $f$ is a positional (resp. bounded) strategy for
Pl.~0 such that $f(0)$ (resp. $f(\pi)$ for any $\pi$) has length less
than $n$, then Pl.~1 has just to start by playing $a_n$ zeros so that
Pl.~1 does not reach the index $a_{n+1}$ and afterwards to complete
the sequence by a finite number of zeros to reach the next index
$a_k$, and so on. Moreover, there does not exist any move-counting
winning strategy $h$ for Pl.~0 because Pl.~1 can start by playing
$a_n$ zeros so that $|h(0,1)|\le n$ and because, at each step $k$,
Pl.~1 can complete the sequence by a finite number of zeros to reach a
new index $a_n$ such that $|h(0,k)|\le n$.
\end{example}

On the other hand, we can show that the existence of a move-counting
winning strategy for Pl.~0 implies $P(W)=1$. The key idea is to
realise that given a move-counting winning strategy $h$, the strategy
$h(\cdot,n)$ is positional.

\begin{proposition}\label{propmove}
Let $\mathcal{G}=(G,v_0,W)$ be a Banach-Mazur game on a finite graph
and $P$ a reasonable probability measure.  If Pl.~0 has a
move-counting winning strategy for $\mathcal{G}$, then $P(W)=1$.
\end{proposition}
\begin{proof}
Let $h$ be a move-counting winning strategy of $Pl.~0$. We denote by
$f_n$ the strategy $h(\cdot,n)$.  Each set \[M_n:=\{\rho\in
\text{Paths}(G,v_0)~:~\rho \ \text{is a play consistent with\ }f_n\}\] has probability~$1$ since $f_n$ is a positional winning
strategy for the Banach-Mazur game $(G,v_0,M_n)$. Moreover, if $\rho$
is a play consistent with $f_n$ for each $n\ge 1$, then $\rho$ is
a play consistent with $h$. In other words, since $h$ is a winning
strategy, we get $\bigcap_n M_n\subset W$. Therefore, as $P(M_n)=1$ for all $n$,
we know that $P(\bigcap_n M_n)=1$ and we conclude that $P(W)=1$.
\end{proof}


Let us notice that the converse of Proposition~\ref{propmove} is false
in general.  Indeed, Example~\ref{ex:pos} exhibit a large set $W$ of
probability~$1$ such that Pl.~0 has no move-counting winning
strategy. However, if $W$ is a countable intersection of
$\omega$-regular sets, then the existence of a winning strategy for
Pl.~0 implies the existence of a move-counting winning strategy for
Pl.~0.


\begin{proposition}\label{countintmovecount}
Let $\mathcal{G}=(G,v_0,W)$ be a Banach-Mazur game on a finite graph
where $W$ is a countable intersection of $\omega$-regular sets
$W_n$. Pl.~0 has a winning strategy if and only if Pl.~0 has a move-counting
winning strategy.
\end{proposition}
\begin{proof}
Let $W=\bigcap_{n\ge 1} W_n$ where $W_n$ is an $\omega$-regular set
and $f$ a winning strategy of Pl.~0 for $\mathcal{G}$.  For any $n\ge 1$, the strategy
$f$ is a winning strategy for the Banach-Mazur game $(G,v_0,W_n)$.
Thanks to~\cite{BGK03}, we therefore know that for any $n\ge 1$, there
exists a positional winning strategy $\tilde{f}_n$ of Pl.~0 for
$(G,v_0,W_n)$.

Let $\phi:\mathbb{N}\to \mathbb{N}$ such that for any $k\ge 1$,
$\{n\in \mathbb{N}:\phi(n)=k\}$ is an infinite\footnote{Such a map $\phi$ exists because one could build a surjection $\psi:\N\to\N\times\N$ and then let $\phi=\psi_1$ where $\psi(n)=(\psi_1(n),\psi_2(n))$.} set.  We consider the
move-counting strategy $h(v,n)=\tilde{f}_{\phi(n)}(v)$. This strategy
is winning because each play $\rho$ consistent with $h$ is a
play consistent with $\tilde{f}_n$ for any $n$ and thus
\begin{align*}
&\{\rho\in \text{Paths}(G,v_0)~:~\rho \ \text{is a play consistent with\ }h\}\\
&\quad\subseteq\bigcap_n \{\rho\in \text{Paths}(G,v_0)~:~\rho \ \text{is a play consistent with\ }\tilde{f}_n\}\\
&\quad\subseteq \bigcap_n W_n=W.
\end{align*}
\end{proof}


\begin{remark}
We cannot extend this result to countable unions of $\omega$-regular
sets because the set of countable unions of $\omega$-regular sets
contains the open sets and Example~\ref{nomove} exhibited a Banach-Mazur
game where $W$ is an open set and Pl.~0 has a winning strategy but no
move-counting winning strategy.
\end{remark}
\begin{remark}\label{cor-omega}
We also notice that if $W$ is a countable intersection of
$\omega$-regular sets, then $W$ is large if and only if $W$ is a set
of probability $1$.  Indeed, the notions of large sets and sets of
probability~$1$ are stable by countable intersection and we know that
a $\omega$-regular set is large if and only if it is of
probability~$1$ \cite{VV06}.
\end{remark}

As a consequence of Remark~\ref{cor-omega}, we have that if $W$ is a
$\omega S$-regular sets, as defined in~\cite{BC06}, the set $W$ is
large if and only if $W$ is a set of probability $1$. Indeed, it is
shown in~\cite{HS12, HST10} that $\omega S$-regular sets are countable
intersection of $\omega$-regular sets. Nevertheless, the following
example shows that, unlike the case of $\omega$-regular sets,
positional strategies are not sufficient for $\omega S$-regular sets.

\begin{example}[\textbf{$\omega S$-regular set with a move-counting winning strategy and without a positional/ bounded  winning strategy}]\label{omegaS}
We consider the complete graph $G_{0,1}$ on $\{0,1\}$ and the set $W$
corresponding to the $\omega S$-regular expression
$((0^*1)^*0^S1)^{\omega}$, which corresponds to the language of words
where the number of consecutive $0$ is unbounded. The move-counting
strategy which consists in playing $n$ consecutive $0$'s at the $n$th
step is winning for Pl.~0. However, clearly enough Pl.~0 does not have
a positional (nor bounded) winning strategy for $W$.
\end{example}

Example~\ref{wwRdef} shows that Remark~\ref{cor-omega} does not extend
to $\omega$-context-free sets. Another notion of simple strategies,
natural inspired by Example~\ref{wwRdef}, is the notion of last-move
strategy. A strategy $f$ for Pl.~0 is said to be \emph{last-move} if
it only depends on the last move of Pl.~1, i.e. for any $v\in V$, for
any $\pi\in\text{FinPaths}(G,v)$, $f(\pi)\in
\text{FinPaths}(G,\text{last}(\pi))$ and a play $\rho$ is consistent
with $f$ if it is of the form $(\pi_i f(\pi_i))_{i \ge 1}$. It is
obvious that there exists a last-move winning strategy for Pl.~0 in
the game described in Example~\ref{wwRdef}. In particular, we deduce
that the existence of a last-move winning strategy for $W$ does not
imply that $W$ has probability $1$. Example~\ref{wwRdef} allows also
us to see that the existence of a last-move winning strategy does not
imply in general the existence of a move-counting winning strategy or
a bounded winning strategy. Indeed, let $W$ be the set
$\{(w_iw^R_i)_i\in \text{Paths}(G_{0,1,2},2):w_i\in\{0,1,2\}^{*}\}$.
Since $P(W)=0$ (and thus $P(W) \ne 1$), we know that Pl.~0 has no
move-counting winning strategy by Proposition~\ref{propmove} and no
bounded winning strategy.


The notion of last-move winning strategy is in fact incomparable with
the notion of move-counting winning strategy and the notion of bounded
winning strategy. Indeed, on the complete graph $G_{0,1}$ on
$\{0,1\}$, if we denote by $W$ the set of runs in $G_{0,1}$ such that
for any $n\ge 1$, the word $1^n$ appears, then Pl.~0 has a
move-counting winning strategy for the game $(G_{0,1},0,W)$ but no
last-move winning strategy. In the same way, if we denote by $W$ the
set of aperiodic runs on $G_{0,1}$ then Pl.~0 has a $1$-bounded
winning strategy for the game $(G_{0,1},0,W)$ but no last-move winning
strategy (it suffices for Pl.~1 to play at each time the same word). 

\section{Generalised Banach-Mazur games}

Let $\mathcal{G}=(G,v_0,W)$ be a Banach-Mazur game on a finite graph.
We know that the existence of a bounded winning strategy or a move-counting
winning strategy of Pl.~0 for $\mathcal{G}$ implies that $P(W)=1$ for every
reasonable probability measure $P$. Nevertheless, it is possible that $P(W)=1$
and Pl.~0 has no bounded winning strategy and no move-counting
winning strategy (Example~\ref{ex:pos}). We therefore search a
new notion of strategy such that the existence of such a winning
strategy implies $P(W)=1$ and the existence of a bounded winning strategy or a move-counting
winning strategy imply the existence of such a winning strategy.
To this end, we introduce a new type of Banach-Mazur games:

\begin{definition}
A \emph{generalised Banach-Mazur game} $\mathcal{G}$ on a finite graph
is a tuple $(G,v_0,\phi_0,\phi_1,W)$ where $G=(V,E)$ is a finite
directed graph where every vertex has a successor, $v_0\in V$ is the
initial state, $W\subset \text{{Paths}}(G,v_0)$, and $\phi_i$ is a map
on $\text{{FinPaths}}(G,v_0)$ such that for any $\pi\in
\text{{FinPaths}}(G,v_0)$, \[\phi_i(\pi)\subset
\mathcal{P}\big(\text{{FinPaths}}(G,\text{{last}}(\pi))\big)\setminus
\{\emptyset\} \ \text{\ and\ } \ \phi_i(\pi)\ne \emptyset.\]
\end{definition}

A generalised\footnote{We only present here a generalisation of Banach-Mazur games on finite graphs but
this generalisation could be extended to Banach-Mazur games on topological spaces by asking
that for any non-empty open set $O$, $\phi_i(O)$ is a collection of non-empty open subsets of $O$.
} Banach-Mazur game $\mathcal{G}=(G,v_0,\phi_0,\phi_1,W)$
on a finite graph is a two-player game where Pl.~0 and Pl.~1 alternate
in choosing \emph{sets of finite paths} as follows: Pl. 1 begins with
choosing a set of finite paths $\Pi_1\in \phi_1(v_0)$; Pl.~0 selects a
finite path $\pi_1\in \Pi_1$ and chooses a set of finite paths
$\Pi_2\in \phi_0(\pi_1)$; Pl 1. then selects $\pi_2\in \Pi_2$ and
proposes a set $\Pi_3\in\phi_1(\pi_1\pi_2)$ and so on. A play of
$\mathcal{G}$ is thus an infinite path $\pi_1\pi_2\pi_3\dots$ in $G$
and we say that Pl. 0 wins if this path belongs to $W$, while Pl.~1
wins if this path does not belong to $W$.

We remark that if we let
$\phi_{\text{ball}}(\pi):=\{\{\pi'\}:\pi'\in\text{FinPaths}(G,\text{last}(\pi))\}$
for any $\pi\in \text{FinPaths}(G,v_0)$, then the generalised
Banach-Mazur game given by
$(G,v_0,\phi_{\text{ball}},\phi_{\text{ball}},W)$ coincides with the
classical Banach-Mazur game $(G,v_0,W)$. We also obtain a game similar
to the classical Banach-Mazur game if we consider the function
$\phi(\pi)=\mathcal{P}(\text{FinPaths}(G,\text{last}(\pi)))$. On the
other hand, if we consider
$\phi(\pi):=\{\{\pi'\}:\pi'\in\text{FinPaths}(G,\text{last}(\pi)),
|\pi'|=1\}$, we obtain the classical games on graphs such as the ones
studied in~\cite{lncs2500}.

We are interested in defining a map $\phi_0$ such that Pl.~0 has a winning strategy for $(G,v_0,\phi_0,\phi_{\text{ball}},W)$ if and only if $P(W)=1$. To this end, we notice that we can restrict actions of Pl.~0 by forcing each set in $\phi_0(\pi)$ to be ``big'' in some sense.
The idea to characterise $P(W)=1$ is therefore to force Pl. 0 to play with finite sets of finite paths of conditional probability bigger than $\alpha$ for some $\alpha >0$.

\begin{definition}
Let $\mathcal{G}=(G,v_0,W)$ be a Banach-Mazur game on a finite graph,
$P$ a reasonable probability measure and $\alpha>0$. An
\emph{$\alpha$-strategy} of Pl.~0 for $\mathcal{G}$ is a strategy of
Pl. 0 for the generalised Banach-Mazur game
$\mathcal{G}_{\alpha}=(G,v_0,\phi_{\alpha},\phi_{\text{ball}},W)$
where \[\phi_{\alpha}(\pi)=\Big\{\Pi\subset\text{{FinPaths}}(G,\text{{last}}(\pi)):
P\Big(\bigcup_{\pi'\in
  \Pi}\text{{Cyl}}(\pi\pi')\Big|\text{{Cyl}}(\pi)\Big)\ge\alpha
\ \text{and}\ \Pi\ \text{is finite}\Big\}.\]
We recall that, given two events $A,B$ with $P(B)>0$, the conditional probability $P(A|B)$ is defined by $P(A|B):=P(A\cap B)/P(B)$.
\end{definition}

We notice that every bounded strategy can be seen as an
$\alpha$-strategy for some $\alpha>0$, since for any $N\ge 1$, there
exists $\alpha>0$ such that for any $\pi$ of length less than $N$, we
have $P(\{\pi\})\ge \alpha$. We can also show that the existence of a 
move-counting winning strategy for Pl.~0 implies the existence of a
winning $\alpha$-strategy for Pl.~0 for every $0<\alpha<1$.

\begin{proposition}\label{movalpha}
Let $\mathcal{G}=(G,v_0,W)$ be a Banach-Mazur game on a finite graph.
If Pl.~0 has a move-counting winning strategy, then Pl.~0 has a winning 
$\alpha$-strategy for every $0<\alpha<1$.
\end{proposition}
\begin{proof}
Let $P$ be a reasonable probability measure, $h$ a move-counting
winning strategy for Pl.~0 and $0<\alpha<1$.  We denote by $g_n$ the
positional strategy defined by
\[g_n(v)=h(v,1)\ h\big(\text{last}(h(v,1)),2\big)\ 
\cdots\ h\big(\text{last}(h(v,1)\ h(\text{last}(h(v,1)),2)\cdots),n\big).\]
Let us notice that the definition of the $g_n$'s implies that for any increasing sequence $(n_k)$, a play
of the form
\begin{equation}
 \pi_1 \ g_{n_1}(\text{last}(\pi_1)) \ \pi_2 \ g_{n_2}(\text{last}(\pi_2))
 \ \cdots \ \pi_k \ g_{n_k}(\text{last}(\pi_k)) \ \cdots
 \label{eqform}
 \end{equation}
  is consistent
 with $h$. Since $g_n$ is a positional strategy, we know that each set
\[M_n:=\{\rho\in
\text{Paths}(G,v_0)~:~\rho \ \text{is a play consistent with
  \ }g_n\}\] has probability~$1$.
  In particular, for any $\pi_0\in \text{FinPaths}(G,v_0)$, we deduce that $P(M_n|\text{Cyl}(\pi_0))=1$. Since
  \[M_n\cap \text{Cyl}(\pi_0)\subseteq  \bigcup_{\pi\in\text{FinPaths}(G,\text{last}(\pi_0))}\text{Cyl}\big(\pi_0\pi g_n(\text{last}(\pi))\big),\]
  we have
  \[P \Big(\bigcup_{\pi\in\text{FinPaths}(G,\text{last}(\pi_0))}\text{Cyl}\big(\pi_0\pi g_n(\text{last}(\pi))\big)\Big|\text{Cyl}(\pi_0)\Big)=1\]
  and since $\text{FinPaths}(G,\text{last}(\pi_0))$ is countable, we deduce that for any $n\ge 1$, any $\pi_0\in \text{FinPaths}(G,v_0)$, there exists a finite subset $\Pi_n(\pi_0)\subset \text{FinPaths}(G,\text{last}(\pi_0))$ such that 
  \[P\Big(\bigcup_{\pi\in \Pi_n(\pi_0)}\text{Cyl}\big(\pi_0\pi g_n(\text{last}(\pi))\big)\Big|\text{Cyl}(\pi_0)\Big)\ge \alpha.\]
   We denote by $\Pi'_n(\pi_0)$ the set $\{\pi g_n(\text{last}(\pi)):\pi\in \Pi_n(\pi_0)\}$ and we let \[f(\pi_0):=\Pi'_{|\pi_0|}(\pi_0).\]
 The above-defined strategy $f$ is therefore a winning
 $\alpha$-strategy for Pl.~0 since each play consistent with $f$ is of the form \eqref{eqform} for some sequence $(n_k)$ and thus consistent with $h$.
\end{proof}

Moreover, the existence of a winning
$\alpha$-strategy for some $\alpha>0$ still implies $P(W)=1$.

\begin{theorem}\label{alphaWun}
Let $\mathcal{G}=(G,v_0,W)$ be a Banach-Mazur game on a finite graph
and $P$ a reasonable probability measure. If Pl. 0 has a winning
$\alpha$-strategy for some $\alpha>0$, then $P(W)=1$.
\end{theorem}
\begin{proof}
Let $f$ be a winning $\alpha$-strategy. We consider an increasing
sequence $(a_n)_{n\ge 1}$ such that for any $n\ge 1$, any $\pi$ of
length $a_n$, each $\pi'\in f(\pi)$ has length less than
$a_{n+1}-a_n$; this is possible because for any $\pi$, $f(\pi)$ is a
finite set by definition of $\alpha$-strategy. Without loss of
generality\footnote{Let $\pi$ be a finite path and $n_\pi \ge \max
  \{|\tau| \text{ such that } \tau \in f(\pi)\}$. One can define
  $\tilde{f}(\pi)$ as the set of finite paths $\sigma$ of length
  $n_\pi$ such that $\tau$ is a prefix of $\sigma$, for some $\tau \in
  f(\pi)$. Given a play $\rho$, one can show that $\rho$ is consistent
  with $f$ if and only if $\rho$ is consistent with $\tilde{f}$.}, we
can even assume that for any $n\ge 1$, any $\pi$ of length $a_n$,
each $\pi'\in f(\pi)$ has exactly length $a_{n+1}-a_n$. We therefore
let \[A:=\{(\sigma_k)_{k\ge 1}\in \text{Paths}(G,v_0): \#\{n:
(\sigma_k)_{a_{n}+1\le k\le a_{n+1}}\in f((\sigma_k)_{1\le k\le
  a_n})\}=\infty\}.\] In other words, $(\sigma_k)_{k\ge 1}\in A$ if
$(\sigma_k)$ can be seen as a play where $f$ has been played on an infinite
number of indices $a_n$. Since $f$ is a winning strategy, $A$ is
included in $W$ and it thus suffices to prove that $P(A)=1$.

We first notice that for any $m\ge 1$, any $n\ge m$, if we let
\[B_{m,n}=\{(\sigma_k)_{k\ge 1}\in \text{Paths}(G,v_0): (\sigma_k)_{a_{j}+1\le k\le a_{j+1}}\notin f((\sigma_k)_{1\le k\le a_j}), \  \forall m\le j\le n\},\]
then $P(B_{m,n})\le (1-\alpha)^{n+1-m}$ as $f$ is an $\alpha$-strategy.
We therefore deduce that for any $m\ge 1$,
\[P\Big(\bigcap_{n=m}^{\infty}B_{m,n}\Big)=0\]
and since $A^c=\bigcup_{m\ge 1}\bigcap_{n=m}^{\infty}B_{m,n}$, we conclude that $P(A)=1$.
\end{proof}

If $W$ is a countable intersection of open sets, we can prove the
converse of Theorem~\ref{alphaWun} and so obtain a characterisation of
sets of probability $1$.

\begin{theorem}\label{pl0}
Let $\mathcal{G}=(G,v_0,W)$ be a Banach-Mazur game on a finite graph
where $W$ is a countable intersection of open sets and $P$ a
reasonable probability measure. Then the following assertions are
equivalent:
\begin{enumerate}
\item $P(W)=1$,
\item Pl. 0 has a winning $\alpha$-strategy for some $\alpha>0$,
\item Pl. 0 has a winning $\alpha$-strategy for all $0<\alpha<1$.
\end{enumerate}
\end{theorem}
\begin{proof}
We have already proved $2.\Rightarrow 1.$, and $3.\Rightarrow 2.$ is
obvious.\\ $1.\Rightarrow 3.$ Let $0<\alpha<1$. Let
$W=\bigcap_{n=1}^{\infty}W_n$ where $W_n$'s are open sets. Since
$P(W)=1$, we deduce that for any $n\ge 1$, $P(W_n)=1$. We can
therefore define a winning $\alpha$-strategy $f$ of Pl.~0 as follows:
if $\text{Cyl}(\pi)\subset \bigcap_{k=1}^{n-1} W_k$ and
$\text{Cyl}(\pi)\not\subset W_n$, we let $f(\pi)$ be a finite set
$\Pi\subset\text{FinPaths}(G,\text{last}(\pi))$ such that
$P\Big(\bigcup_{\pi'\in
  \Pi}\text{Cyl}(\pi\pi')|\text{Cyl}(\pi)\Big)\ge\alpha$ and for any
$\pi'\in \Pi$, $\text{Cyl}(\pi\pi')\subset W_n$. Such a finite set
$\Pi$ exists because $W_n$ has probability
$1$ and $W_n$ is an open set, i.e. a countable union of cylinders. This
concludes the proof.
\end{proof}
\begin{remark}
We cannot hope to generalise the latter result to any set $W$. More
precisely, there exist sets of probability~$1$ for which no winning
$\alpha$-strategy exists. Indeed, given a set $W$, on the one hand,
the existence of a winning $\alpha$-strategy for $W$ implies the
existence of a winning strategy for $W$, and thus in particular such a
$W$ is large. On the other hand, we know that there exists some meagre
(in particular not large) set of probability~$1$ (see
Example~\ref{wwRdef}).  However, one can ask whether the existence
of a winning $\alpha$-strategy is equivalent to the fact that $W$ is a
large set of probability $1$.
\end{remark}
When $W$ is a countable intersection of open sets, 
we remark that the generalised Banach-Mazur game
$\mathcal{G}_{\alpha}=(G,v_0,\phi_{\alpha},\phi_{\text{ball}},W)$ is
in fact determined.

\begin{theorem}\label{thm-long}
Let $\mathcal{G}_{\alpha}$ be the generalised Banach-Mazur game 
given by $\mathcal{G}_{\alpha}=(G,v_0,\phi_{\alpha},\phi_{\text{ball}},W)$
where $G$ is a finite graph, $W$ is a countable intersection of open sets
 and $P$ a reasonable probability measure. Then the following assertions are
equivalent:
\begin{enumerate}
\item $P(W)<1$,
\item Pl. 1 has a winning strategy for $\mathcal{G}_{\alpha}$ for some
  $\alpha>0$,
\item Pl. 1 has a winning strategy for $\mathcal{G}_{\alpha}$ for all
  $0<\alpha<1$.
\end{enumerate}
\end{theorem}
\begin{proof}
We deduce from Theorem~\ref{pl0} that $2.\Rightarrow 1.$ because
$\mathcal{G}_{\alpha}$ is a zero-sum game, and $3.\Rightarrow 2.$ is
obvious.\\ $1.\Rightarrow 3.$ Let $W=\cap_{n=1}^{\infty}W_n$ with
$P(W)<1$ and $W_n$ open.  We know that there exists $n\ge 1$ such that
$P(W_n)<1$. It then suffices to prove that Pl.~1 has a winning
strategy for the generalised Banach-Mazur game
$(G,v_0,\phi_{\alpha},\phi_{\text{ball}},W_n)$ for all
$0<\alpha<1$. Without loss of generality, we can thus assume that $W$
is an open set. We recall that $W$ is open if and only if it is a
countable union of cylinders. Since any strategy of Pl.~1 is winning
if $W=\emptyset$, we also suppose that $W\ne \emptyset$.

Let $0<\alpha<1$. We first show that there exists a finite path $\pi_1\in \text{FinPaths}(G,v_0)$ such that any set $\Pi_2\in \phi_{\alpha}(\pi_1)$ contains a finite path $\pi_2$ satisfying 
\begin{equation}P(W|\text{Cyl}(\pi_1\pi_2))\le P(W)<1.\label{cond}
\end{equation}

Let \begin{equation}
I_W:=\inf\{P(W|\text{Cyl}(\pi)):\pi\in \text{FinPaths}(G,v_0)\}.
\label{Iw}
\end{equation} Since $W$ is a non-empty union of cylinders,
there exists $\sigma\in \text{FinPaths}(G,v_0)$ such that
${P(W|\text{Cyl}(\sigma))=1}$. We remark that $P(W)=\sum_{\pi:|\pi|=|\sigma|}P(W|\text{Cyl}(\pi))P(\text{Cyl}(\pi))$ and
$\sum_{\pi:|\pi|=|\sigma|}P(\text{Cyl}(\pi))=1$. Therefore, since $P(W|\text{Cyl}(\sigma))>P(W)$, we deduce that there exists
$\pi\in \text{FinPaths}(G,v_0)$ with $|\pi|=|\sigma|$ such that ${P(W|\text{Cyl}(\pi))<P(W)}$.
We conclude that $I_W<P(W)$ and thus, by definition of $I_W$, there
exists $\pi_1 \in \text{FinPaths}(G,v_0)$ such that
\begin{equation}
I_W+\frac{1}{\alpha}(P(W|\text{Cyl}(\pi_1))-I_W)<P(W).
\label{cond I}
\end{equation}
Let $\Pi_2\in \phi_{\alpha}(\pi_1)$. We consider $\tau_1,\dots,\tau_n\in \Pi_2$ and $\sigma_1,\dots,\sigma_m\in \text{FinPaths}(G,\text{last}(\pi_1))$ such that cylinders $\text{Cyl}(\tau_i)$, $\text{Cyl}(\sigma_j)$ are pairwise disjoint,
$\bigcup_{\pi\in \Pi_2}\text{Cyl}(\pi)\subset \bigcup_{i=1}^n \text{Cyl}(\tau_i)$ and 
\begin{equation}
\text{Paths}(G,\text{last}(\pi_1))=\bigcup_{i=1}^n \text{Cyl}(\tau_i)\cup \bigcup_{j=1}^m \text{Cyl}(\sigma_j).
\label{tausig}
\end{equation}
 Assume that for all $1 \le i\le n$, we have 
 \begin{equation}
 P(W|\text{Cyl}(\pi_1\tau_i))> P(W).
 \label{cont}
 \end{equation}
 Then, we get
\begin{align*}
&P(W|\text{Cyl}(\pi_1))\\
\quad&= \sum_{i=1}^n P(W\cap\text{Cyl}(\pi_1\tau_i)|\text{Cyl}(\pi_1)) + \sum_{j=1}^m P(W\cap\text{Cyl}(\pi_1\sigma_j)|\text{Cyl}(\pi_1))\ \ \text{by disjointness and \eqref{tausig}}\\
\quad&= \sum_{i=1}^nP(W|\text{Cyl}(\pi_1\tau_i))P(\text{Cyl}(\pi_1\tau_i)|\text{Cyl}(\pi_1)) + \sum_{j=1}^m P(W|\text{Cyl}(\pi_1\sigma_j))P(\text{Cyl}(\pi_1\sigma_j)|\text{Cyl}(\pi_1))\\
\quad&\ge P(W)\sum_{i=1}^nP(\text{Cyl}(\pi_1\tau_i)|\text{Cyl}(\pi_1))+I_W\sum_{j=1}^m P(\text{Cyl}(\pi_1\sigma_j)|\text{Cyl}(\pi_1))\ \ \text{by \eqref{cont} and \eqref{Iw}}\\
\quad&\ge P(W)\sum_{i=1}^nP(\text{Cyl}(\pi_1\tau_i)|\text{Cyl}(\pi_1))+I_W\big(1-\sum_{i=1}^nP(\text{Cyl}(\pi_1\tau_i)|\text{Cyl}(\pi_1))\big)\ \ \text{by \eqref{tausig}}\\
\quad&\ge P(W)P\big(\bigcup_{\pi\in \Pi_2}\text{Cyl}(\pi_1\pi)|\text{Cyl}(\pi_1)\big)+I_W\big(1-P\big(\bigcup_{\pi\in \Pi_2}\text{Cyl}(\pi_1\pi)|\text{Cyl}(\pi_1)\big)\big)\ \ \text{by properties of $\tau_i$'s}\\
\quad&\ge P(W)\alpha + I_W(1-\alpha)  \quad\text{(because $\Pi_2\in \phi_{\alpha}(\pi_1)$ and $P(W)>I_W$)}
\end{align*}
and thus
$P(W)\le I_W+\frac{1}{\alpha}(P(W|\text{Cyl}(\pi_1))-I_W)$
which is a contradiction with \eqref{cond I}. We conclude that if $\pi_1$ is given by $\eqref{cond I}$, then any set $\Pi_2\in \phi_{\alpha}(\pi_1)$ contains a finite path $\pi_2$ satisfying \eqref{cond}.

We can now exhibit a winning strategy for Pl.~1.  We assume that Pl.~1
begins with playing a finite path $\pi_1$ satisfying \eqref{cond I}. Let
$f$ be an $\alpha$-strategy. We know that Pl.~1 can select a finite
path $\pi_2\in f(\pi_1)$ satisfying \eqref{cond},
i.e. $P(W|\text{Cyl}(\pi_1\pi_2))\le P(W)$. By repeating the above
method from $\pi_1\pi_2$, we also deduce the existence of a finite
path $\pi_3$ such that any set $\Pi_4\in
\phi_{\alpha}(\pi_1\pi_2\pi_3)$ contains a finite path $\pi_4$
satisfying $P(W|\text{Cyl}(\pi_1\pi_2\pi_3\pi_4))\le P(W)$. We can
thus assume that Pl.~1 plays such a finite path $\pi_3$ and then
selects $\pi_4\in f(\pi_1\pi_2\pi_3)$ such that
$P(W|\text{Cyl}(\pi_1\pi_2\pi_3\pi_4))\le P(W)$. This strategy is a
winning strategy for Pl.~1. Indeed, as $W$ is an open set and thus a
countable union of cylinders, if
$P(W|\text{Cyl}(\pi_1\cdots\pi_{2n}))\le P(W)<1$ for any $n$, then
$\pi_1\pi_2\pi_3\cdots\notin W$.
\end{proof}
\begin{corollary}
Let $0<\alpha<1$. The generalised Banach-Mazur game $\mathcal{G}_{\alpha}=(G,v_0,\phi_{\alpha},\phi_{\text{ball}},W)$ is determined when $W$ is a countable intersection of open sets. More precisely, 
Pl.~0 has a winning strategy for $\mathcal{G}_{\alpha}$ if and only if $P(W)=1$, and
Pl.~1 has a winning strategy for $\mathcal{G}_{\alpha}$ if and only if $P(W)<1$.
\end{corollary}

Since the existence of a bounded winning strategy for Pl.~0 implies
the existence of a winning $\alpha$-strategy for Pl.~0 and 
the existence of a move-counting winning strategy for Pl.~0 implies
the existence of a winning $\alpha$-strategy for Pl.~0,
we deduce from Example~\ref{nobound} and Example~\ref{nomove} that
in general, the existence of a winning $\alpha$-strategy for Pl.~0 does not
imply the existence of a move-counting winning strategy Pl.~0
and the existence of a bounded winning strategy for Pl.~0.
 On the other hand, we know that there exists a Banach-Mazur game for which Pl.~0 has a bounded
 winning strategy and no last-move winning strategy. The existence of a winning $\alpha$-strategy thus does
  not imply in general the existence of a last-move winning strategy.
Conversely, if we consider the game $(G_{0,1},0,W)$ described in
Example~\ref{wwRdef}, Pl.~0 has a last-move winning strategy but no
winning $\alpha$-strategy (as $P(W)=0$). The notion of $\alpha$-strategy is thus
incomparable with the notion of last-move strategy.


\section{More on simple strategies}

We finish this paper by considering the crossings between the
different notions of simple strategies and the notion of bounded
strategy i.e. the bounded length-counting strategies, the bounded
move-counting strategies and the bounded last-move
strategies. Obviously, the existence of a bounded length-counting winning strategy for Pl.~0
implies the existence of a length-counting winning strategy for Pl.~0, and
we have this implication for each notion of bounded strategies and their
no bounded counterpart. We start by noticing that the existence of a
bounded move-counting winning strategy is equivalent to the existence
of a positional winning strategy.

\begin{proposition}\label{moreon}
Let $\mathcal{G}=(G,v_0,W)$ be a Banach-Mazur game on a finite graph.
Pl.~0 has a bounded move-counting winning strategy if and only if
Pl.~0 has a positional winning strategy.
\end{proposition}
\begin{proof}
Let $h$ be a bounded move-counting winning strategy for Pl.~0.
We denote by $C_1,\dots,C_N$ the bottom strongly connected components (BSCC) of $G$.
Let $1\le i\le N$.
Since $h$ is a bounded strategy and $G$ is finite, there exist some finite paths
$w^{(i)}_1,\dots,w^{(i)}_{k_i}\subset C_i$ such that for any $v\in C_i$, for any $n\ge 1$,
\[h(v,n)\in\{w^{(i)}_1,\dots,w^{(i)}_{k_i}\}.\]
Let $v\in V$. If $v\in C_i$, we let $f(v)=\sigma_0w^{(i)}_1\sigma_1w^{(i)}_2\sigma_2\dots w^{(i)}_{k_i}$ where $\sigma_l$ are finite paths in $C_i$ such that $f(v)$ is a finite path in $C_i$ starting from $v$. If $v\notin \bigcup_i C_i$, we let $f(v)=\sigma_v$ where $\sigma_v$ starts from $v$ and leads into a BSCC of $G$. The positional strategy $f$ is therefore winning
as each play $\rho$ consistent with $f$ can be seen as a play consistent with $h$.
\end{proof}

The other notions of bounded strategies are not equivalent to any other notion of simple strategy.

\begin{example}[\textbf{Set with a bounded length-counting winning strategy and without a positional winning strategy}]
Let $G_{0,1}$ be the complete graph on $\{0,1\}$, $(\rho_n)$ an
enumeration of finite words in $\{0,1\}$ and
$\rho_{\text{target}}=0\rho_1\rho_2\cdots$. We consider the set
$W=\{\sigma\in\{0,1\}^{\omega}:\#\{i\ge
1:\sigma(i)=\rho_{\text{target}}(i)\}=\infty\}$. It is evident that
Pl.~0 has a bounded length-counting winning strategy for the game
$(G_{0,1},0,W)$. However, Pl.~0 has no positional winning
strategy. Indeed, if $f$ is a positional strategy such that
$f(0)=a(1)\cdots a(k)$, then Pl.~1 can play according to the strategy
$h$ defined by $h(\sigma(1)\cdots
\sigma(n))=\sigma(n+1)\cdots\sigma(N)\ 0$ such that for any $n+1\le
i\le N$, $\sigma(i)\ne \rho_{\text{target}}(i)$,
$\rho_{\text{target}}(N+1)\ne 0$ and for any $1\le i\le k$, $a(i)\ne
\rho_{\text{target}}(N+i+1)$.
\end{example}

\begin{example}[\textbf{Set with a bounded last-move winning strategy and without a positional winning strategy}]
Let $G_{0,1,2}$ be the complete graph on $\{0,1,2\}$. For any $\phi:\{0,1,2\}^*\to \{0,1\}$,
if we consider the set $W:=\{(\pi_i\phi(\pi_i))_{i \ge 1}: \pi_i\in\{0,1,2\}^*\}$, then Pl.~0
has a 1-bounded last-move winning strategy given by $\phi$ for the game $(G_{0,1,2},2,W)$.
On the other hand, we can choose $\phi$ such that Pl.~0 has no positional winning
strategy. Indeed, it suffices to choose $\phi:\{0,1,2\}^*\to \{0,1\}$ such that for any
$\pi\in\{0,1,2\}^*$, any $n\ge 1$, any
$\sigma(1),\dots,\sigma(n)\in\{0,1,2\}$, there exists $k\ge 1$ such
that $\phi(\pi 2^k)\ne \sigma(1)$ and for any $1\le i\le n-1$,
$\phi(\pi 2^k\sigma(1)\cdots\sigma(i))\ne \sigma(i+1)$.  Such a
function exists because the set $\{0,1,2\}^*$ is countable. 
Therefore, Pl.~0 has no positional winning strategy
for the game $(G_{0,1,2},2,W)$ because, if $f$ is a positional strategy and
$f(2)=\sigma(1)\dots \sigma(n)$, then Pl.~1 can play consistent with the strategy
$h$ defined by $h(\pi)=2^k$ such that $\phi(\pi 2^k)\ne \sigma(1)$ and for
any $1\le i\le n-1$, $\phi(\pi 2^k\sigma(1)\cdots \sigma(i))\ne \sigma(i+1)$.
Pl.~0 has thus a 1-bounded last-move winning strategy and no positional winning strategy
for the game $(G_{0,1,2},2,W)$.
\end{example}

\begin{example}[\textbf{Set with a bounded winning strategy and without a bounded length-counting winning strategy}]
Let $G_{0,1,2,3}$ be the complete graph on $\{0,1,2,3\}$.
For any $\phi:\{0,1,2,3\}^*\to\{0,1\}$, if we denote by $W$ the set of runs $\rho$
such that
${\#\{n\ge1:\phi(\rho(1)\dots\rho(n))=\rho(n+1)\}=\infty}$,
then Pl.~0 has a 1-bounded winning strategy given by $\phi$ for the game $(G_{0,1,2,3},2,W)$.
We now show how we can define $\phi$ so that Pl.~0 has no bounded 
length-counting winning strategy.
Let $n_k=\sum_{i=1}^k 3i$. We choose
${\phi:\{0,1,2,3\}^*\to\{0,1\}}$ such that for any $k\ge 1$, any
$\pi\in\{0,1,2,3\}^*$ of length $n_{k}$ and any
$\sigma(1),\dots,\sigma(k)\in\{0,1,2,3\}$, there exists $\tau\in
\{2,3\}^*$ of length $2k$ such that $\phi(\pi\tau \, 2)\ne \sigma(1)$
and for any $1\le i\le k-1$, $\phi(\pi\tau \, 2 \,
\sigma(1)\cdots\sigma(i))\ne\sigma(i+1)$. Such a function exists
because the cardinality of $\{2,3\}^{2k}$ is equal to the cardinality
of $\{0,1,2,3\}^k$ and the length of $\pi\tau \, 2 \,
\sigma(1)\cdots\sigma(k)<n_{k+1}$.  Therefore, Pl.~0 has no
bounded length-counting winning strategy because if $f$ is a $k$-bounded
length-counting strategy (for some $k \in \mathbb{N}$) and
$f(2,n_k+k+1)=\sigma$, then Pl.~1 can start by playing $2^{n_k} \tau \, 2$,
where $\tau\in \{2,3\}^*$ of length $2k$ such that $\phi(\pi\tau \,
2)\ne \sigma(1)$ and for any $1\le i\le k-1$, $\phi(\pi\tau \, 2 \,
\sigma(1)\cdots\sigma(i))\ne\sigma(i+1)$, and if Pl.~1 keep playing
with same philosophy, then Pl.~1 wins the play.
Pl.~0 has thus a 1-bounded winning strategy and no bounded length-counting winning strategy
for the game $(G_{0,1,2},2,W)$.
\end{example}

\medskip

The relations between the different notions of simple strategies on a
finite graph can be summarised as depicted in Figure~\ref{resume2}.
\begin{figure}[h!]
    \begin{center}
      \begin{tikzpicture}[xscale=1,yscale=1.5,>=stealth',shorten >=1pt]
       
        \path (0,6) node[] (q00) {Bounded Positional};
        \path (0,5) node[] (q0) {Positional};
        \path (0,4) node[] (q1f) {Finite memory};

        \path (-7.5,6) node[] (q21) {$1$-Bounded};
        \path (-7.5,5.25) node[] (q22) {$2$-Bounded};
        \path (-7.5,4.5) node[] (qdots) {$\vdots$};
        \path (-7.5,4) node[] (q2b) {$b$-Bounded};
        \path (-7.5,3.25) node[] (qdots) {$\vdots$};
        \path (-7.5,2) node[] (q2) {Bounded};
        \path (-7.5,1) node[] (q2alpha) {$\alpha$-strategy};
        \path (4,1.1) node[] (q2lm) {Last-move};
        \path (0,2.75) node[] (q1) {\begin{tabular}{c}Bounded\\ move-counting\end{tabular}};
        \path (-4.1,2.25) node[] (q1blc) {\begin{tabular}{c}Bounded\\ length-counting\end{tabular}};
        \path (4,2.25) node[] (q1blm) {\begin{tabular}{c}Bounded\\ last-move\end{tabular}};

        \path (0,1.5) node[] (q3) {Move-counting};

        \path (0,0) node[] (q4) {Length-counting};
        \path (0,-1) node[] (q5) {General};


        \draw[double,arrows=->] (q1blm) .. controls +(160:4.55cm) .. (q2);

        \draw[double,arrows=<->] (q1f) -- (q0);
        \draw[double,arrows=<->] (q0) -- (q00);
        \draw[double,arrows=<->] (q1f) -- (q1);
        \draw[double,arrows=->] (q1) -- (q1blc);
        \draw[double,arrows=->] (q1) -- (q1blm);
        \draw[double,arrows=->]  (q1blc) -- (q2);
        \draw[double,arrows=->] (q21) -- (q22);
        \draw[double,arrows=->] (q22) -- (-7.5,4.7);
        \draw[double,arrows=->] (q2b) -- (-7.5,3.45);
        \draw[double,arrows=->] (-7.5,2.75) -- (q2);

        \draw[double,arrows=->] (q1blm) -- (q2lm);
        \draw[double,arrows=->] (q2lm) -- (q5);

        \draw[double,arrows=->] (q1) -- (q3);
        \draw[double,arrows=->] (q2) -- (q2alpha);
        \draw[double,arrows=->] (q2alpha) -- (q5);
        \draw[double,arrows=->] (q3) -- (q4);
        \draw[double,arrows=<->] (q4) -- (q5);

        \draw[double,arrows=->] (q3) -- (q2alpha);

        \path (-8.5,.25) node[right] (qqq) {proba1};
        \path (-8.5,-.25) node[right] (qqq) {notproba1};
        \draw[thick,dashed] (-8.5,0)--(-4,0) -- (5,1.8);
      \end{tikzpicture}
    \end{center}
\caption{Winning strategies for Player~0 on finite graphs.}
\label{resume2}
\end{figure}
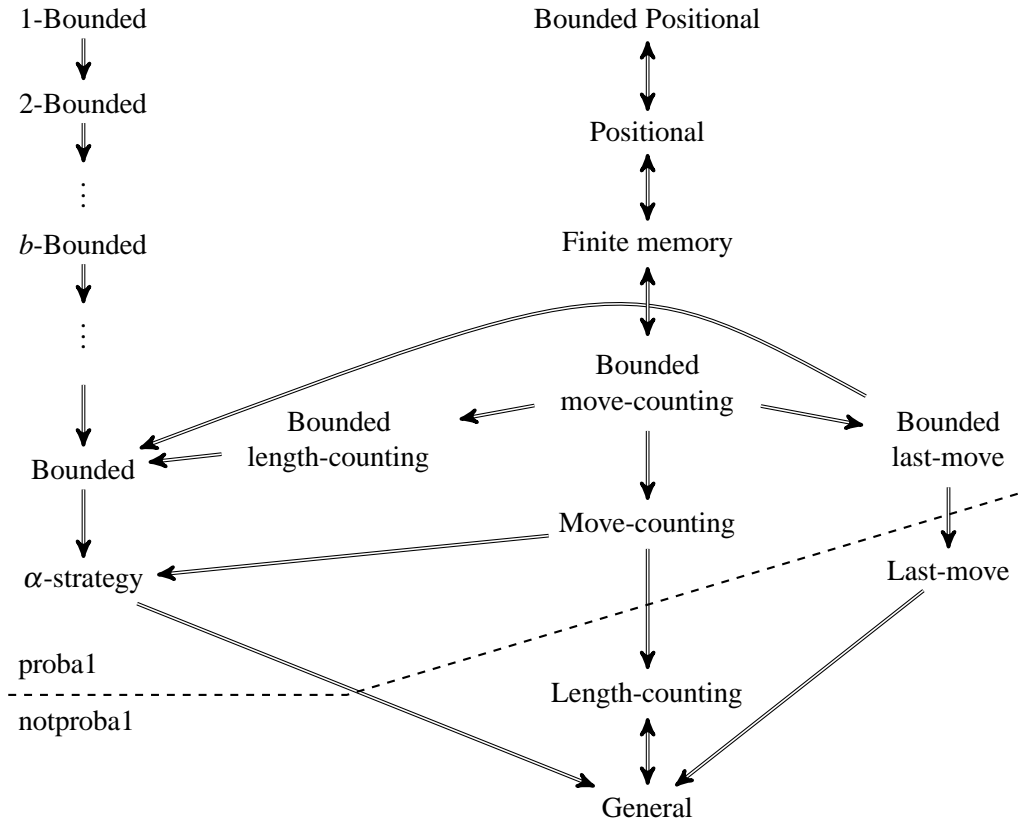
We draw attention to the fact that the situation is very different in
the case of infinite graphs. For example, a positional strategy can be
unbounded, the notion of length-counting winning strategy is not
equivalent to the notion of winning strategy (except if the graph is
finitely branching), and the notion of bounded move-counting winning
strategy for Pl.~0 is not equivalent to the notion of positional
winning strategy.

\begin{example}[\textbf{Set on an infinite graph with a bounded move-counting winning strategy and without a positional winning strategy}]
We consider the complete graph $G_{\mathbb{N}}$ on $\mathbb{N}$ and
the game ${\mathcal{G}=(G_{\mathbb{N}},0,W)}$ where
$W=\{(\sigma_k)\in \mathbb{N}^\omega: \forall~n\ge 1,\ \exists~k\ge
1,\ (\sigma_k,\sigma_{k+1})=(n,n+1)\}$.  Pl.~0 has a bounded
move-counting winning strategy given by $h(v,n)=n~n+1$ but no
positional winning strategy.
\end{example}

\bibliographystyle{eptcs}
\bibliography{biblio}



\end{document}